\documentclass[12pt]{article}
\usepackage{amsmath}
\usepackage{amssymb,amsthm}
\usepackage{color}
\usepackage{graphicx}
\UseRawInputEncoding
\newtheorem{Thm}{Theorem}[section]
\newtheorem{theorem}[Thm]{Theorem}
\newtheorem{proposition}[Thm]{Proposition}

\newtheorem{remark}{Remark}[section]

\newtheorem{definition}[Thm]{Definition}

\title{Dynamics of states of infinite quantum systems as a cornerstone of the second law of thermodynamics}

\author{Walter F. Wreszinski\footnote{wreszins@gmail.com, 
Instituto de Fisica, Universidade de S\~ao Paulo (USP), Brasil}}        
        
\begin{document}

\maketitle

\begin{abstract}


We improve on our version of the second law of thermodynamics as a deterministic theorem for quantum spin systems \cite{Wre} in two basic
aspects. The first concerns the general statement of the second law: spontaneous changes in an adiabatically closed system will always be
in the direction of increasing (mean) entropy, which rises to a maximal value. In order to arrive at this statement, we have to extend 
the previous definition of adiabatic transformation, in order to include sudden interactions, and thereby generalize the ``barrier model''. 
The cornerstone of the second law is then seen to be the fact that the dynamics may induce a basic structural transition between states of 
infinite systems in the limit of large times, which, physically, represent times much larger than specific relaxation times. Two specific 
examples concern the transition from pure to mixed states in two different universality classes of dynamics in one dimension, one being the 
exponential model, which has exactly soluble dynamics and does not display a phase transition, the other the Dyson model(s), for which 
neither the dynamics nor the statistical mechanics are exactly soluble, and which do exhibit a (ferromagnetic) phase transition, 
first demonstrated by Dyson. It is also shown, as a consequence of results of Albert and Kiessling on the Cloitre function, that there is 
strong graphical evidence that the dynamics of the Dyson models are chaotic for large times, and, thereby, that the mechanisms of approach 
to equilibrium differ for these two universality classes in a profound way, which depend exclusively on the dynamics, and not on the states 
and observables.      

\end{abstract}

\section{Introduction: states and observables for classical systems, chaotic states, and the smoothing property of densities}

In his fundamental work \cite{Di}, Dirac emphasized that the basic building blocks of quantum mechanics are \textbf{states} and \textbf{observables},
states being vectors in a (finite-dimensional) Hilbert space, and observables, self-adjoint operators acting on these states. In the present
paper we shall see that it is imperative, in order to arrive at a precise statement of the second law of thermodynamics in the sense of
Clausius, to generalize this notion of state to infinite dimensions, or, more precisely, to infinite number of degrees of freedom (ndof). In
order to sharpen our intuition, it appears to be useful to undertake a preliminary investigation of the situation in connection with
classical dynamical systems, which leads us to the probabilistic properties of (deterministic) dynamical systems \cite{LaMa}. Since generic
dynamical systems are chaotic, we choose to look at one of the simplest chaotic systems, the so-called dyadic map \cite{LaMa}. 

Define $T_{2}$  by $T_{2} : X \to Y$ , where $X = [0, 1]$ and $Y = [0, 1]$, by
\begin{equation}
\label{(1.1)}
T_{2}x = fr(2x) \equiv 2x \mbox{mod}  1
\end{equation}
Above, $fr(x)\equiv (x-[x])$, where $[x]$  denotes the largest integer which is smaller or equal to x. Note that this mapping is not one-to-one 
(it is a so-called endomorphism). In fact, we see that the inverse image of a point $x$  is either $x/2$ or $\frac{x+1}{2}$. $T_{2}$ is called the 
dyadic transformation and leaves Lebesgue measure $\mu$ (on the line) invariant because the inverse image of a point $y$ is
\begin{equation}
\label{(1.2)}
T_{2}^{-1}(y) = \{\frac{y}{2}\} \cup \{\frac{y+1}{2}\}
\end{equation}
Indeed, from \eqref{(1.2)}, $\mu(T_{2}^{-1}([0, y])) = y$, which generalizes to
\begin{equation}
\label{(1.3)}
\mu(A) = \mu(T_{2}^{-1}(A)) \mbox{ for any Borel set } A \in [0, 1]
\end{equation}
$T_{2}$ has a simple alternative description. Let $x$ be given
by its expansion in basis 2, i.e.,
\begin{equation}
\label{(1.4)}
x = .\epsilon_{1}\epsilon_{2} \cdots = \sum_{n=1}^{\infty} \frac{\epsilon_{n}}{2^{n}} \mbox{ with } \epsilon_{n} \in \{0,1\}
\end{equation}
Then, it is easy to see that
\begin{equation}
\label{(1.5)}
T_{2}x = .\epsilon_{2}\epsilon_{3} \cdots \mbox{ and, in general } T_{2}^{n}x = .\epsilon_{n+1}\epsilon_{n+2} \cdots
 \end{equation}
that is, $T_{2}$ is the ``one-sided shift'' in this representation.
Given a Lebesgue integrable function, i.e.,
\begin{equation}
\label{(1.6)}
f \in  L^{1}(0, 1)
\end{equation}
we may define the Ruelle-Perron-Frobenius operator $P$ (\cite{LaMa}) by
\begin{equation}
\label{(1.7)}
\int_{A}(P f )(x)dx = \int_{T_{2}^{-1}(A)} f(x) dx
\end{equation}
If
\begin{equation}
\label{(1.8)}
f_{0}(x) \equiv 1
\end{equation}
it follows that
\begin{equation}
\label{(1.9)}
Pf_{0} \equiv  1
\end{equation}
Let $A = [0, x]$. From \eqref{(1.7)}, it follows that
\begin{eqnarray*} 
(Pf)(x) = \frac{d}{dx}(\int_{0}^{\frac{x}{2}}f(u)du + \int_{\frac{x}{2}}^{\frac{x+1}{2}} f(u)du =\\
= \frac{1}{2}[f(\frac{x}{2})+f(\frac{x+1}{2})]
\end{eqnarray*}
from which, by iteration,
\begin{equation}
\label{(1.10)}
(P^{n}f)(x) = \frac{1}{2^{n}}\sum_{k=0}^{2^{n}-1} f(\frac{x+k}{2^ {n}})
\end{equation}
We shall say that f is a density if $f\ge 0 $ almost everywhere in $(0,1)$ 
$\int_{0}^{1} f(x) dx = 1$; $f_{0}$, given by \eqref{(1.8)}, is the uniform density.
The words a.e. stand for almost everywhere, that is, in the complement of a set with (Lebesgue) measure zero.
By \eqref{(1.7)}, P maps densities to densities, and, by \eqref{(1.10)},
\begin{equation}
\label{(1.11)}
\lim_{n \to \infty} (P^{n}(f))(x) = \int_{0}^{1} f(y) dy = 1
\end{equation}
Together with \eqref{(1.8)}, \eqref{(1.11)} is a form of the property of approach to equilibrium  (see also \cite{LaMa}, 
Theorem 4.4.1, p.65): the evolution of any density tends (for (discrete) time $n \to \infty$) to the unique invariant density, the uniform
density $f_{0}$ given by \eqref{(1.8)}.

Of particular importance above is that $P$ is not defined on individual orbits ot the map, which would correspond
to taking $f$ in \eqref{(1.7)} to be a ``delta function'', which is not allowed by condition \eqref{(1.6)}. Indeed, these individual orbits
behave rather erratically. Consider two points $x_{1} , x_{2}$, both in [0, 1], close in the sense that the first n digits in
the expansion \eqref{(1.9)} are identical: it follows that $T_{2}x_{1}$ and $T_{2}x_{2}$ differ already in the first digit: 
an initially exponentially small difference $2^{-n}$ is magnified by the evolution to one of order O(1). When this property holds
for $n$ arbitrarily large, as in the present example, one speaks of the \textbf{exponential sensitivity to initial conditions}.

Exponential sensitivity to initial conditions is a precise definition of the \textbf{chaotic} element of the dynamics: for almost all
$x_{0}$ (in the sense of Lebesgue measure, i.e., excluding a set of zero Lebesgue measure) $\gamma$, $T_{2}^{n}x_{0}$ comes arbitrarily
close to almost any $x \in [0,1]$, if $n$ is taken sufficiently large, or, in other words, it ``fills'' the whole interval uniformly
throughout the evolution. Therefore $\lim_{n \to \infty} T_{2}^{n}x_{0}$ \textbf{does not exist for a.e. $x_{0}$}. The set $\gamma$ consists
of the finite dyadic numbers $x_{f}$, i.e., those whose dyadic expansion ~\eqref{(1.4)} is finite; it is immediate from the definition
that $T_{2}^{n}x_{f} \to 0 \mbox{ or } 1$, the latter being the fixed points of the dyadic map: they are untypical in the sense that
they do not fill the interval uniformly.

We now come to the subject of \textbf{states and observables for the dyadic map and chaotic states}. For an alternative treatment of
states of classical systems, see Ruelle's article \cite{Ru}. See also the paper by Illner and Neunzehrt \cite{IN}, in which the change
between the two topologies as below is held responsible for the fact that in the limit of a process with time-reversal invariance in the
Boltzmann equation a ``time-arrow'' is generated. In order to characterize these two topologies, we begin with a definition.

\begin{definition}
\label{Definition 1.1}
A density $f$ is a function $f \ge 0$ a.e. in $(0,1)$ such that $\int_{0}^{1} f(x)dx = 1$. $f_{0}$, given by ~\eqref{(1.8)}, is
the \textbf{uniform density}.
\end{definition}

The present model may strengthen our intuition behind our choices of states and observables. Indeed, taking $f$ as a bounded continuous
function on $[0,1]$, as an observable, and, as a state, the delta measure
\begin{equation}
\label{(1.12)}
\omega_{1}(f) \equiv \delta_{x}(f) \mbox{ for } x \in [0,1]
\end{equation}
then the ``time-evolution'' leads to
\begin{equation}
\label{(1.13)}
\omega_{1,n}(f) = \delta_{x_{n}}(f) = f(x_{n})
\end{equation} 
It follows, however, that, for a.e. $x_{0}$, $x_{n}=T_{2}^{n}x_{0}$ comes arbitrarily close to almost any $x \in [0,1]$ \cite{LaMa}. Thus,
$f(x_{n})$ cannot converge (for a.e. $x_{0}$), and \textbf{there is no approach to equilibrium}.

Consider, now, as observables, the \textbf{densities} defined in ~\eqref{Definition 1.1}. Define, now, as observable the entropy $\eta$
by $\eta(u) \equiv -u \log(u) \mbox{ with } \eta(0) \equiv 0$, and consider the states
\begin{equation}
\label{(1.14)}
\omega_{2}(f) \equiv \int_{0}^{1} dx \eta(f)(x)
\end{equation}
By ~\eqref{(1.14)}, together with Theorem 9.2.1 and remark 9.2.1 of \cite{LaMa}, the sequence of ``time-evolved states'', where $P$
denotes the Ruelle-Perron-Frobenius operator ~\eqref{(1.7)},
\begin{equation}
\label{(1.15)}
\omega_{2,n}(f) \equiv \int_{0}^{1} dx \eta(P^{n}(f))(x)
\end{equation}
with $n=1.2, \cdots$, is monotone increasing in $n$, and bounded above by zero, and therefore converges. Thus: 
\textbf{there is approach to equilibrium}! This is the \textbf{smoothing property of densities}. It is due to the fact that densities
require, for their definition, a (non-countably) infinite set of points, i.e., they must be supported on a set of non-zero Lebesgue
measure. Equation ~\eqref{(1.15)} is the analogue of the forthcoming definition of the mean entropy as a functional of states of
\textbf{infinite} systems in quantum statistical mechanics.

We may finally observe, for the purpose of further comparison with the quantum case, that the state $\omega_{2}$ above is manifestly
\textbf{not} invariant under ``time-reversal'', since the dynamics, specified by $P$, is not. In fact, definitions ~\eqref{(1.14)} of
$\omega_{2}$ and ~\eqref{(1.7)} of $P$ show that, for suitable ``initial'' observables $f$, $\omega_{2}(f)$ exhibits exponential
sensitivity to initial conditions, due to the fact that the dynamics $T_{2}$ has this property: we say that $\omega_{2}(f)$ is a
\textbf{chaotic state}.

The present model seems thus to suggest that the main issue in the problem of irreversibility and approach to equilibrium in quantum
mechanics may also lie in the choice of suitable states and observables, in particular their \textbf{non-invariance} under time-reversal. 
We now examine this issue in greater detail.

\section{Dynamics of states and the Schr\"{o}dinger paradox. The case of states of infinite systems}

In the present year, which marks the centennary of (Schr\"{o}dinger) quantum mechanics, it may be specially significant to handle a topic
to which Schr\"{o}dinger provided, in 1935, the initial crucial input: the \textbf{Schr\"{o}dinger paradox} \cite{Schr}. This topic is the theory
of irreversibility and the second law of thermodynamics, already treated in \cite{Wre}, but in an incomplete way. In this paper we 
attempt to supply the missing steps and, at the same time, clarify the principles underlying the choice between the von Neumann and the
Boltzmann entropy, which, to our knowledge, has not been done before.

\subsection{Models of qss and their dynamical properties}

A typical example of the models we shall use is the generalized Heisenberg Hamiltonian (generalized Ising model (gIm) if $J_{1}=0$)
for any finite region $\Lambda \subset \mathbf{Z}^{\nu}$
\begin{equation}
\label{(2.1)}
H_{\Lambda} = -2\sum_{x,y \in \Lambda}[J_{1}(x-y)(S_{x}^{1}S_{y}^{1}+S_{x}^{2}S_{y}^{2})+J_{2}(x-y)S_{x}^{3}S_{y}^{3}]
\end{equation}
where
\begin{equation}
\label{(2.2)}
\sum_{x \in \mathbf{Z}^{\nu}}|J_{i}(x)|< \infty \mbox{ and } J_{i}(0)=0 \mbox{ for } i=1,2
\end{equation}
Above, $S_{x}^{i}=1/2 \sigma_{x}^{i}, i=1,2,3$ and $\sigma_{x}$ are the
Pauli matrices at the site $x$. Above, $H_{\Lambda}$ acts on the Hilbert space ${\cal H}_{\Lambda}=\otimes_{x \in \Lambda}\mathbf{C}_{x}^{2}$.
We define
${\cal A}(\Lambda) = B({\cal H}_{\Lambda})$
These local algebras satisfy:
1.) Causality: $[{\cal A}(B),{\cal A}(C)]=0$ if $B \cap C = \phi$ and 2.) Isotony: $B \subset C \Rightarrow {\cal A}(B) \subset {\cal A}(C)$.
${\cal A}_{L} = \cup_{B} {\cal A}(B)$ is termed the \emph{local} algebra; its closure with respect to the norm, the \emph{quasilocal} algebra
(observables which are, to arbitrary accuracy, approximated by observables attached to a \emph{finite} region). The norm is
$A \in B({\cal H}_{\Lambda}) \to ||A|| = sup_{||\Psi|| \le 1} ||A \Psi||$, $\Psi \in {\cal H}_{\Lambda}$.
An automorphism $\tau_{t}(A)$ equals the norm limit of $\lim_{\Lambda \nearrow \infty} \exp(iH_{\Lambda}t) A \exp(-iH_{\Lambda}t)$ 
for $A \in {\cal A}(\Lambda)$. The limit $\lim_{\Lambda \nearrow \infty}$ will denote the van Hove limit \cite{BRo2}, p.287).

The notion of state generalizes to systems with
infinite number of degrees of freedom $\omega(A)= \lim_{\Lambda \nearrow \infty} \omega_{\Lambda}(A)$, at first for $A \in {\cal A}_{L}$
and then to ${\cal A}$.

Each state $\omega$ defines a representation $\Pi_{\omega}$ of ${\cal A}$ as bounded operators on a Hilbert space ${\cal H}_{\omega}$ with
cyclic vector $\Omega_{\omega}$ (i.e., $\Pi_{\omega}({\cal A}) \Omega_{\omega}$ is dense in ${\cal H}_{\omega}$), such that 
$\omega(A) = (\Omega_{\omega}, \Pi_{\omega}(A) \Omega_{\omega})$ (the GNS construction). The strong closure of $\Pi_{\omega}({\cal A})$ is
a von Neumann algebra, with commutant $\Pi_{\omega}({\cal A})^{'}$, which is the set of bounded operators on ${\cal H}_{\omega}$ which
commute with all $\Pi_{\omega}({\cal A})$, and the center is defined by $Z_{\omega}= \Pi_{\omega}({\cal A}) \cap \Pi_{\omega}({\cal A})^{'}$.

Considering quantum spin systems on $\mathbf{Z}^{\nu}$, we shall consider only space-translation-invariant states, i.e., such that
\begin{equation}
\label{(2.3)}
\omega \circ \tau_{x} = \omega \mbox{ for all } x \in \mathbf{Z}^{\nu}
\end{equation} 

For states of finite ndof, we assume periodic boundary conditions.

A \textbf{pure}, or ergodic state is a state which cannot be
written as a proper convex combination of two distinct states $\omega_{1}$ and $\omega_{2}$, i.e., the following does \emph{not}
hold:
\begin{equation}
\label{(2.4)}
\omega = \alpha \omega_{1} + (1-\alpha) \omega_{2} \mbox{ with } 0<\alpha<1
\end{equation}
If the above formula is true, it is natural to regard $\omega$ as a mixture of two pure ``phases''
$\omega_{1}$ and $\omega_{2}$, with proportions $\alpha$ and $1-\alpha$, respectively.

Let $\Gamma = \mathbf{Z}^{\nu}$ and $P_{0}(\Gamma)$ denote the collection of all finite parts of $\Gamma$. We may generalize ~\eqref{(2.1)}
to quantum spin systems described by a Hamiltonian for any finite region $\Lambda \subset \mathbf{Z}^{\nu}$ given by
\begin{equation}
\label{(2.5)}
H_{\Phi}(\Lambda) = \sum_{X \subset \Lambda} \Phi(X)
\end{equation}
where $\Phi$ is a translation-invariant interaction function from $P_{0}(\Gamma)$ into the Hermitian elements of the quasi-local
algebra ${\cal A}$, such that $\Phi(X) \in {\cal A}(X)$ and
\begin{equation}
\label{(2.6)}
\Phi(X+x) = \tau_{x}(\Phi(X)) \forall X \subset \mathbf{Z}^{\nu}, \forall x \in \mathbf{Z}^{\nu}
\end{equation}
We shall for simplicity restrict ourselves to quantum spin systems with two-body interactions. In this case, stability requires that
\begin{equation}
\label{(2.7)}
\sum_{0 \ne x \in \mathbf{Z}^{\nu}} ||\Phi(\{0,x\})|| < \infty
\end{equation}

A prototype of the above is the \textbf{anisotropic ferromagnetic Heisenberg model}: it corresponds to the choices:
\begin{equation}
\label{(2.8)} 
J_{1} < 0 \mbox{ and } J_{2} < 0 \mbox{ with } |J_{1}| < |J_{2}| \mbox{ and } J_{1}(x)=J_{2}(x)=0 \mbox{ except if } x,y \mbox{ are n.n. }
\end{equation}
where $n.n.$ denote nearest-neighbors in the lattice $\mathbf{Z}^{\nu}$. If $J_{2} \equiv 0$, we speak of \textbf{the X-Y model}. The assumption 
$J_{1} \equiv 0$ in ~\eqref{(2.1)}, with $J_{2}=J$ satisfying ~\eqref{(2.2)}, defines the \textbf{generalized Ising model} (gIm) studied by Emch 
\cite{Em} and Radin \cite{Ra}: 
\begin{equation}
\label{(2.9)}
H_{\Lambda} = \frac{1}{2} \sum_{x,y \in \Lambda} J(|x-y|) \sigma_{x}^{3}\sigma_{y}^{3}
\end{equation}
Let, now, $\nu=1$, and define two subclasses of models of the gIm:
\begin{equation}
\label{(2.10)}
J(|x|) = -\xi^{-|x|} \mbox{ with } \xi > 1 \mbox{ and } x \ne 0 (\mbox{ \emph{exponential model} } E_{\xi})
\end{equation}
and
\begin{equation}
\label{(2.11)}
J(|x|) = -\frac{1}{|x|^{\alpha}} \mbox{ with } \alpha > 1 \mbox{ and } x \ne 0 (\mbox{ \emph{Dyson model} } D_{\alpha})
\end{equation}
The above conditions on $\xi$ and $\alpha$ are required for stability \eqref{(2.5)}. Under certain conditions, the Dyson
model displays a ferromagnetic phase transition \cite{Dy}. 
With $P_{0}(\mathbf{Z}^{\nu})$ denoting the set of all finite subsets of $\mathbf{Z}^{\nu}$ as before, for each triple 
$A=(A_{1},A_{2},A_{3})$, where the $A_{i} \in P_{0}(\mathbf{Z}^{\nu}, i=1,2,3$ are pairwise disjoint, define
\begin{equation}
\label{(2.12)}
\sigma^{A} \equiv \prod_{x_{1}\in A_{1}} \sigma_{x_{1}}^{1} \prod_{x_{2}\in A_{2}} \sigma_{x_{2}}^{2} \prod_{x_{3}\in A_{3}} \sigma_{x_{3}}^{3} 
\end{equation}
Let $\omega$ be any state satisfying
\begin{equation}
\label{(2.13)}
\omega(\sigma^{A}) = 0 \forall A \mbox{ such that } A_{3} \ne \phi
\end{equation}

We have the following:

\begin{theorem}
\label{th:2.1}
For the gIm with interaction either of the exponential model $E_{\xi}$ with $\xi$ a transcendental number, or the Dyson
model $D_{\alpha}$, and any initial state (at $t=0$) $\omega$ satisfying \eqref{(2.13)},
\begin{equation}
\label{(2.14)}
\lim_{t \to \infty} (\omega \circ \tau_{t}) = \omega_{eq} \equiv \otimes_{x \in \mathbf{Z}} (2^{-1}Tr_{x}) 
 \end{equation}
where the limit on the l.h.s. of \eqref{(2.14)} is taken in the weak* topology.
\end{theorem}

The equilibrium state above is the \textbf{tracial state}, defined, for general dimension $\nu$, as
\begin{equation}
\label{(2.15)}
\omega_{tr} \equiv \otimes_{x \in \mathbf{Z}^{\nu}} (2^{-1}Tr_{x})
\end{equation}
where the infinite sum over $x$ is short for the limit in the weak*-topology of the set of sums over finite regions. 

For the proof, see (\cite{Ra}, Corollary, p.2953).

Quantum spin systems are similar to quantum fields because of the now famous Lieb-Robinson bound (\cite{LR}, see also \cite{BRo2}, p.254).
Following the last reference, let us assume, instead of \eqref{(2.5)},
\begin{equation}
\label{(2.16)}
||\Phi||_{\lambda} \equiv \sum_{x \in \mathbf{Z}^{\nu}} ||\Phi(\{0,x\})|| \exp(\lambda |x|) < \infty
\end{equation}
for some $\lambda>0$. Then, for $A,B \in {\cal A}_{0}$,
\begin{equation}
\label{(2.17)}
||[(\tau_{x}\tau_{t})(A),B]|| \le 2||A||||B|| \exp[-|t|(\lambda |x|/|t|-2 ||\Phi||_{\lambda})]
\end{equation}
In \eqref{(2.17)},$(\tau_{x}\tau_{t})$ may be replaced by $(\tau_{t}\tau_{x})$. The commutator $[(\tau_{x}\tau_{t})(A),B]$ with
$B \in {\cal A}_{0}$ provides a measure of the dependence of the observation $B$ at the point $x$ at time $t$ at the origin at time $t=0$,
showing that this effect decreases exponentially with time outside the cone
$$
|x| < |t| (\frac{2||\Phi||_{\lambda}}{\lambda}) 
$$
Equation \eqref{(2.17)} means that physical disturbances propagate with ``group velocity'' bounded by
\begin{equation}
\label{(2.18)}
v_{\Phi} \equiv \inf_{\lambda} (\frac{2||\Phi||_{\lambda}}{\lambda})
\end{equation}

\subsection{The mean entropy, the Schr\"{o}dinger paradox and the time-arrow: the necessity of states of infinite systems}

For a finite quantum spin system the von Neumann entropy is defined as ($k_{B}=1$):
\begin{equation}
\label{(2.19)}
S_{\Lambda} = -Tr (\rho_{\Lambda} \log \rho_{\Lambda})
\end{equation}
We may view $\rho_{\Lambda}$ as a \emph{state} $\omega_{\Lambda}$ on ${\cal A}(\Lambda)$ - a positive, normed linear functional on ${\cal A}(\Lambda)$:
$\omega_{\Lambda}(A) = Tr_{{\cal H}_{\Lambda}} (\rho_{\Lambda} A) \mbox{ for } A \in {\cal A}(\Lambda)$
(positive means $ \omega_{\Lambda}(A^{\dag}A) \ge 0$, normed $\omega_{\Lambda}(\mathbf{1})=1$.)

We may correspondingly view $S_{\Lambda}$ as a \textbf{functional} of the state $\omega_{\Lambda}$. Accordingly, $S_{\Lambda}(t)$ will be a functional 
$S_{\Lambda}(\omega_{\Lambda,t})$ of the time-evolved state 
\begin{equation}
\label{(2.20)}
\omega_{\Lambda,t} \equiv Tr_{{\cal H}_{\Lambda}} (\rho_{\Lambda} \tau_{t}^{\Lambda}(A))
\end{equation}
where
\begin{equation}
\label{(2.21)}
\tau_{t}^{\Lambda}(A) \equiv \exp(iH_{\Lambda}t) A \exp(-iH_{\Lambda}t) \forall  A \in {\cal A}(\Lambda)
\end{equation}
By ~\eqref{(2.19)} and ~\eqref{(2.20)} and the cyclicity and invariance of the trace under unitary transformations, it follows that
\begin{equation}
\label{(2.22)}
\omega_{\Lambda,t} = \omega_{\Lambda,-t}
\end{equation}

We now come back to the second law. Assume, as we shall always do in this paper, that we are dealing with a \textbf{closed} system,
i.e., a system completely isolated from all external influences. This idealization is the starting point of traditional thermodynamics
\cite{tHW}, according to which the second law of thermodynamics, in the formulation of Clausius (\cite{tHW}, p.66) may be stated:

\textbf{Second law (Clausius)}
\textbf{Spontaneous changes in an adiabatically closed system will always be in the direction of increasing entropy, which rises to a maximum}.

We shall refer to the above statement as the dynamical formulation of the second law, assuming, as it happens in Nature, that it refers
to ``spontaneous changes in time''. Given that the von Neumann entropy is a functional of the time-evolved state, we may express the
dynamical formulation of the second law as
\begin{equation}
\label{(2.23)}
S_{\Lambda}(\omega_{\Lambda,t_{1}}) \le S_{\Lambda}(\omega_{\Lambda,t_{2}}) \mbox{ if } t_{1}<t_{2}
\end{equation} 
Since, however, $-t_{2} < -t_{1}$, ~\eqref{(2.6)} yields
\begin{equation}
\label{(2.24)}
S_{\Lambda}(\omega_{\Lambda,t_{2}}) \le S_{\Lambda}(\omega_{\Lambda,t_{1}})
\end{equation}
Equations ~\eqref{(2.23)} and ~\eqref{(2.24)} yield, together:
\begin{equation}
\label{(2.25)}
S_{\Lambda}(\omega_{\Lambda,t_{2}})  =  S_{\Lambda}(\omega_{\Lambda,t_{1}})
\end{equation}
Equation ~\eqref{(2.25)} is the expression of \textbf{Schr\"{o}dinger's paradox}, see \cite{Schr} and Lebowitz's inspiring review \cite{Leb}.

The reason why we have been so explicit in the above statement is that a variant of it also applies to infinite systems, as we now explain.
We have seen that the notion of state generalizes to systems with
infinite number of degrees of freedom $\omega(A)= \lim_{\Lambda \nearrow \infty} \omega_{\Lambda}(A)$, at first for $A \in {\cal A}_{L}$
and then to ${\cal A}$.
The state $\omega_{\Lambda}^{t}(A) = \omega_{\Lambda}(\exp(iH_{\Lambda}t)A\exp(-iH_{\Lambda}t))$ does \emph{not} have a limit
as $\Lambda \nearrow \infty$ because the spectrum is discrete and the state an almost periodic function of $t$.
The functional $S_{\Lambda}$ is a continuous functional of the state $\omega_{\Lambda,t}$, and, thus, also a quasi-periodic function
of $t$. As we shall see now, for states of infinite systems, however, the situation is \emph{entirely different}!

For a large system the \emph{mean entropy} is the natural quantity from the physical standpoint:
\begin{equation}
\label{(2.26)}
s(\omega) \equiv \lim_{\Lambda \nearrow \infty} (\frac{S_{\Lambda}}{|\Lambda|})(\omega)
\end{equation}
The mean entropy has the following two properties \cite{LanRo}:
\begin{itemize}
\item [$a.)$] $0 \le s(\omega) \le \log D$ where $D=2S+1$;
\item [$b.)$] $s$ is \emph {upper semicontinuous}, that is $lim sup_{n \to \infty} s(\omega_{n}) \le s(\omega)$.
\end{itemize}
In b.), $\omega_{n}$ is a sequence of states such that $\omega_{n} \to \omega$ 
in the weak*- topology, i.e., $\omega_{n} (A) \to \omega(A) \forall A \in {\cal A}$. 

We shall exclusively use the weak*-topology in the present paper. In the uniform or norm topolgy, the mean entropy is continuous, by
a theorem of Fannes \cite{Fa}, but there is no physical requirement which would reinforce the extra condition of uniformity defining
the norm topology. We may call this a condition of \textbf{physical naturalness}: it is very important, and, indeed, crucial. Our
switching to ``densities'', such as the mean entropy, is also of this nature, because these are the quantities of physical relevance-
the ``measurable quantities''. 

One simple example of b.) is a characteristic function of a closed set. The mean 
entropy's basic property is the property of \textbf{affinity} (\cite{LanRo}, \cite{BRo2}):
\begin{equation}
\label{(2.27)}
s(\alpha \omega_{1} + (1-\alpha) \omega_{2}) = \alpha s(\omega_{1}) + (1-\alpha) s(\omega_{2}) \mbox{ with } 0\le \alpha \le 1
\end{equation}

We now come back to the problem of the ``arrow of time''.

According to Penrose \cite{Pe}, the following is a list of phenomena associated to the ``arrow of time'': ``the second law of thermodynamics,
the expansion of the Universe,. the use of retarded potentials in electrodynamics, the decay of the $K^{0}$-meson, and our subjective 
experience that the time passes''.

Schr\"{o}dinger's paradox is also one of these phenomena, because, being a consequence of the time-reversal invariance of the dynamical 
laws, it can only be overcome by the \textbf{existence} of a time-arrow, which, as in the example of the dyadic map, may be associated to
the existence of the mathematical limit of infinite time for densities, whose physical significance, explained in \cite{Wre}, is of 
a universal time, much greater than the typical relaxation times of systems (in a given universality class). This is very similar to the situation in 
statistical mechanics, where mathematical discontinuities in the thermodynamic functions arise \textbf{only} in the thermodynamic limit. We thus
propose the following modified version of the second law:

\textbf{Modified Second law (Clausius)}
\textbf{The mean entropy of an adiabatically closed system rises monotonically to its maximum value}.

Strict uppersemicontinuity of the mean entropy, as in item b.) above requires the infinite time limit : the limit $n \to \infty$ in 
item b.) above is short for the limit $t_{n} \to \infty$, where $t_{n}, n=1,2, \cdots$ represents \textbf{any} sequence of time values 
tending to infinity.

\begin{proposition}
\label{prop:2.2}
Let the general interaction potential in ~\eqref{(2.13)} satisfy 
\begin{equation}
\label{(2.28)}
|\Phi(\{0,x\})| < C (1+|x|)^{-\alpha} \mbox{ with } \alpha > 2\nu
\end{equation}
Then, modified second law (Clausius) does not hold, unless, either the mean entropy remains constant in time, or a time-arrow exists.

\begin{proof}
By Theorems 2.6 and 3.1 of \cite{Wre}, see also \cite{NSY}, under the condition ~\eqref{(2.28)}, for any two fixed values of the time 
$t_{n_{1}}$ and $t_{n_{2}}$,
\begin{equation}
\label{(2.29)}
s(\omega_{t_{n_{1}}}) = s(\omega_{t_{n_{2}}}) = s(\omega_{t_{0}}) \equiv s_{0}
\end{equation}

The above equation is the precise infinite-volume analogue of ~\eqref{(2.23)}, which expresses Schr\"{o}dinger's paradox. If strict
upper semicontinuity is assumed in item b.) above, together with time-reversal invariance, we obtain
\begin{equation}
\label{(2.30)}
s(\omega_{-\infty}) > s_{0} 
\end{equation}
with the interpretation of $s(\omega_{-\infty})$ as the mean entropy at the ``time-reversed'' universal time, contradicting the monotonic
increase of the mean entropy assumed in modified second law (Clausius).
\end{proof}
\end{proposition} 

\subsection{Two different ways to introduce a time arrow: 1) The Boltzmann approach; 2) Adiabatic transformations and the Lieb-Yngvason approach}

We have seen in ~\eqref{prop:2.2} that replacement of the entropy of a system of finite ndof by the mean entropy (a ``density'') does not
per se overcome the Schr\"{o}dinger paradox: a time-arrow must be assumed. As it happens, however, currently accepted physical theories
(with the major exception of those describing the expansion of the Universe, see the epilogue) are time-reversal invariant. Moreover, the
corresponding dynamics provide a good description of \textbf{equilibrium states}: KMS states for $T>0$, ground states for $T=0$, see section 2.4.1.

Although it was Boltzmann who discovered (from ``nothing'') the notion of entropy, and the Gibbs-von Neumann entropy is a generalization
of his $``\log(\rho_{\Lambda})''$ version which matches the ``microcanonical ensemble'', he was certainly keenly aware of the Schr\"{o}dinger
paradox (\cite{Leb}), which led him to introduce a new notion of entropy, today known as the Boltzmann entropy (see \cite{Leb} for
a very illuminating introduction to this concept). It was most clearly stated and explained by Narnhofer \cite{Na} in her recent preprint
that, actually, the Boltzmann equation is an Ansatz for a dynamical evolution, and, under reasonable hypotheses \cite{Na}, equilibrium (KMS)  
states are globally stable. There are three main aspects of the Boltzmann approach: 1) the Boltzmann time-evolution already
contains a time-arrow, and therefore does not \textbf{per se} require the concept of adiabatic transformation; 2.) the Boltzmann 
time-evolution does not provide (and does not pretend to provide) a description of equilibrium states; and 3.) the (mean) entropy 
does not, in general, necessarily rise to a maximal value.

The Boltzmann entropy is closely related \cite{PPP} to the R.A. Fisher information, which was proved to be monotonic for the Boltzmann
equation in \cite{ISV}. 

In contrast, our proof of the modified second law (Clausius) in the forthcoming section 3 relies heavily on the concept of adiabatic 
transformation (see section 2.4.2), which is a generalization of the barrier model, because its so-called ``preparation part'' introduces
a time-arrow. In this sense, we are closer in spirit to the approach of Lieb and Yngvason \cite{LY}, who rely (heavily) on their notion
of ``adiabatic accessibility''. For systems with finite ndof their work provides a complete, axiomatic approach to the second law, in the
form of a precise version of the Kelvin-Planck statement. We attempt, in section 3, to provide a dynamic complement to their work; many
of their basic intuitive ideas have been, however, incorporated into our treatment. In addition, since we do not change the dynamics, (initial) 
equilibria continue to be well described by the associated KMS and stability conditions (see section 2.4.1), and final states are 
non-equilibrium steady states (NESS), associated to maximal mean entropy and ``maximal mixedness''.

\subsection{Dynamical properties of states of infinite systems related to the modified second law}

\subsubsection{Equilibrium properties, basic structural transitions underlying the modified second law and Boltzmann's dictum}

Our choice of the initial state $\omega_{0}$ as an equilibrium state may also be expected to be expressible in terms of some properties
of the dynamics, viz., of the function $t \to \omega_{t}$. For temperature $T>0$ this is the well-known KMS boundary condition, see \cite{Hu}. We
shall restrict ourselves in section 3 to the case $T=0$, and the analogous \textbf{ground state condition} is also well-known
(\cite{Hu}, Theorem 3.2): there is a function $F(z)$, analytic for $Im z <0$, uniformly bounded for $Im z \le 0$, continuous on the real
axis, such that, for $z=t$ (real)
\begin{equation}
\label{(2.31)}
F(t) = \omega(\tau_{t}(A)B) \forall A,B \in {\cal A}
\end{equation}
Let $H$ be the GNS Hamiltonian associated to the state $\omega$. We say that $\omega$ satisfies the \textbf{stability condition} iff
\begin{equation}
\label{(2.32)}
H \ge 0
\end{equation}
on the associated GNS space. By (\cite{Hu}, Theorem 3.4), a state which satisfies the ground state condition is invariant and satisfies
the stability condition, and vice-versa.

An initial equilibrium state may undergo, in the limit of infinite time, a basic structural change, for instance from a pure to a mixed state,
or from a mixed state to a ``maximally mixed state''. Note that such transitions may occur only between states of infinite systems. We shall
henceforth always assume that we have to do with an initial state $\omega_{0}$ such that $\omega_{t_{n}} \to \omega$ in the weak*-topology, where
$t_{n}, n=1,2, \cdots$ is any sequence of time-values tending to infinity.      

\begin{proposition}
\label{prop:2.4.1}
Assume ~\eqref{(2.28)}. If the initial state $\omega_{0}$ is a convex combination of a fixed number of pure states, and the final state is
the tracial state ~\eqref{(2.15)}, then modified second law (Clausius) holds in the strict sense. Furthermore, each of the finite-volume approximants
of the final state is the ``most mixed state'' in the sense of Lemma 2.2 of \cite{NT}. Their limit (in the weak*-topology) is also the 
``most mixed state'', in the sense of being associated to maximal mean entropy. Finally, the tracial state incorporates in a natural way
Boltzmann's ``maximum wealth of microstates''.

\begin{proof}

If the initial state is a convex combination of a fixed number of pure states, its mean entropy is zero due to the property of affinity 
~\eqref{(2.27)}, and the facts 1.) that the mean entropy of a pure state is zero and 2.) the mean entropy of the tracial state equals the maximum 
value $\log 2$. Facts 1.) and 2.) follow under condition ~\eqref{(2.28)} by theorems 2.6 and 3.1 of \cite{Wre}, see also \cite{NSY}.   

Lemma 2.2 of \cite{NT} asserts that, for a density matrix $\rho$, if one measures ``mixedness'' by the quantity $Tr(\rho(1-\rho))= 1-D(0,\rho)$,
where, for any two operators $A$ and $B$, $D(A,B) \equiv Tr((A^{*}-B^{*})(A-B))$ and $A^{*}$ denotes the Hermitian conjugate operator, then the
quantity $D(0,\rho)$ is a strictly convex functional of $\rho$. Indeed, for $0 \le \alpha,\beta \le 1$, the Schwarz inequality yields
$D(\alpha \rho_{1} + \beta \rho_{2}) \le (\alpha^{2}+\alpha)Tr(\rho_{1}^{2})+(\beta^{2}+\beta)Tr(\rho_{2}^{2})$. As a consequence, the ``mixedness''
of a state grows with the number of states in a convex combination, and is attained at the ``most mixed state'', the tracial state, which is
the finite-volume version of ~\eqref{(2.15)}. Although the proof holds only for finite number of ndof, if one measures ``mixedness'' of a state 
$\omega$ by the value $s(\omega)$, the tracial state ~\eqref{(2.15)} is, again, the ``most mixed state'', because it exhibits the highest value 
$\log 2$ of the mean entropy.
Consider, now, the set of states
$$
\omega_{\Lambda_{1},\Lambda_{2},\Lambda_{3}} \equiv \otimes_{x \in \Lambda_{3}}[1/2(\omega_{+,x}+\omega_{-,x})] \otimes \omega_{1,2}
$$
where
$$
\omega_{1,2} \equiv \otimes_{x \in \Lambda_{1}} \omega_{+,x} \otimes_{x \in \Lambda_{2}} \omega_{-,x}
$$
with $\Lambda_{1} \cup \Lambda_{2} \cup \Lambda_{3} = \Lambda$ and $\omega_{\pm,x} \equiv  _{x}(\pm|, \dot |\pm)_{x} \forall x \in \mathbf{Z}^{\nu}$.
Above, $|\pm)_{x}$ are the eigenvectors of $\sigma^{3}_{x}$, and $x$ is any vector in $\mathbf{Z}^{\nu}$. Consider the limit, if it exists
in the weak*-topology, of 
$$
\lim_{\Lambda \nearrow \mathbf{Z}^{\nu}} \omega_{\Lambda_{1},\Lambda_{2},\Lambda_{3}} = \tilde{\omega}
$$
We have seen that the ``most probable state'' is $\tilde{\omega} = \omega_{tr}$, the tracial state, if we ``measure'' the probability by 
the mean entropy. If we call ``microstates'' the states $|\pm)_{x}$, the ``most probable state'' corresponds to the choices 
$|\Lambda_{3}| = O(|\Lambda|)$, together with $|\Lambda_{1}| = o(|\Lambda|$ and $|\Lambda_{2}| = o(|\Lambda|$. This may be proved under
the condition ~\eqref{(2.28)} and is a precise formulation of \textbf{Boltzmann's dictum} that the ``most probable state'' corresponds to the
``maximum wealth of microstates''. 
\end{proof} 
\end{proposition}

The above proposition shows that basic structural changes such as $\mbox{ pure } \to \mbox{ mixed }$, or $\mbox{ mixed } \to \mbox{ most mixed }$ are
the cornerstones of the modified second law (Clausius). On the other hand, by ~\eqref{prop:2.2}, the dynamics must also allow a ``time-arrow''.
This leads us to our next subject.

\subsubsection{Adiabatic transformations as generalizations of the barrier model. Sudden interactions}

Concerning irreversible processes, Lieb and Yngvason assert in their basic paper \cite{LY} (p.35): ``the existence of many such processes lies
at the heart of thermodynamics. If they did not exist, it means that nothing is forbidden, and hence there would be no second law''.

In  this connection, an immediate question arises (see also \cite{LY}, Theorem 2.9, ``Caratheodory's principle and irreversible processes).
Caratheodory's principle asserts that ``Any thermodynamical state of a system has neighborhoods which cannot be reached by adiabatic processes''
(\cite{tHW}, p.33). This principle seems, thus, to assert that, for a process or transformation, the implication adiabatic $\Rightarrow$ reversible
is true, or, equivalently, irreversible $\Rightarrow$ non-adiabatic.

But, in connection with the last implication, what about the free irreversible adiabatic expansion of a gas (\cite{tHW}, 2.4.2, p.27)?

Switching for a moment to conventional thermodynamics, with $dQ$ denoting the ``heat exchange'' (a concept criticized by Lieb and Yngvason for
good reason), if $dQ=0$ is posed as a ``definition'' of adiabatic transformation, the formula
\begin{equation}
\label{(2.33)}
\frac{dQ}{T} = dS
\end{equation}
``defining'' the entropy $S$ seems to imply that $dS=0$ (reversibility) if $T \ne 0$ throughout. But, in ~\eqref{(2.33)}, $Q$ is not a function of
state (it depends on the process). The same happens to $T$, which, in ~\eqref{(2.33)} is not a number, but a function which depends on the process-
as is very well-known! Indeed, a counterexample to the implication irreversible $\Rightarrow$ non-adiabatic is provided by the following
\begin{remark}
\label{Remark 2.4.1}
\textbf{The barrier model}
We consider $N$ particles of an ideal gas, initially inside a container of volume $V$, which is allowed to expand adiabatically into another
container of the same volume $V$ (for simplicity), in which, originally, vacuum had been established (Gay-Lussac experiment, \cite{tHW}, p.36).
The change
$$
\Delta S \equiv S(U,2V,N) - S(U,V,N)
$$
of the entropy $S$, with the internal energy $U$ fixed (it is a constant of the motion) is found to be \cite{tHW}
$$
\Delta S = kN \log(2) \mbox{ or, equivalently } \Delta s = k \log(2)
$$
In order to obtain the above formula, it was assumed that the gas will fill the second container uniformly, when equilibrium is attained.
Although not mentioned in \cite{tHW}, this assumption depends crucially on two points: a.) after removal of the barrier, a \textbf{new}
equilibrium state is established for the system; b.) this state must be compatible with the \textbf{ab-initio} assumptions on the states
of the theory, corresponding to a given universality class. In the present case, \textbf{translation-invariance} ~\eqref{(2.3)} determines
the final configuration. Since we are dealing with a finite system, this assumption is replaced by that of periodic boundary conditions, 
as mentioned after ~\eqref{(2.3)}.
\end{remark}

A generalization of the barrier model leads to the concept of \textbf{adiabatic transformations}. We have seen that they may be irreversible,
in agreement with second law (Clausius), but it remains to find a proper generalization of the concept, which is both physically natural
and introduces a mathematically precise concept of ``sudden perturbations''. Somewhat surprisingly, the physics literature on such perturbations
seems to be very scarce, an exception being Pauli's text \cite{Pa}. For systems with finite number of ndof, Lieb and Yngvason solved this problem 
by introducing the alternative concept of ``adiabatic accessibility'' \cite{LY}.
One further aspect is important in our generalization of the barrier model: the \textbf{initial state}. It is an equilibrium state in our
tratment of the barrier model, which evolves into a non-equilibrium state by the ``preparation part'' of the adiabatic transformation: this
is the ``new'' equilibrium state in part a.) above, after removal of the barrier. In the original barrier model, this new non-equilibrium state
instantaneously evolves into the final state, which occupies both containers uniformly, as being the only configuration satisfying the 
ab-initio conditions imposed on the states. Our generalization removes the instantaneous transition, allowing for a class of perturbations
taking place within a non-zero time interval.

A formulation of the concept of adiabatic transformation along the lines we present was first done by Penrose \cite{Pe}, but without the 
concept of sudden interactions.

\begin{definition} 
\label{Definition 2.4.2}
Let a C*-dynamical system $({\cal A}, \tau_{t})$ be given, with $t \in [-r,\infty)$. An \emph{adiabatic transformation} of a given (initial) state
$\omega_{in}$ consists of three successive steps. The first (resp. third) steps are dynamical automorphic evolutions of the state $\omega_{in}$ (for
$t < -r$) (resp. $\omega_{0}$ for $t > 0$), of the form
\begin{equation}
\label{(2.34)}
t \in \mathbf{R} \to \omega_{t} \equiv \omega \circ \tau_{t}
\end{equation}
where the circle denotes composition, i.e., $(\omega \circ \tau_{t})(A)=\omega(\tau_{t}(A))$, and $\tau$ denotes a given automorphism group
under which $\omega_{in}$ is invariant. The \textbf{second step}, from $-r \le t \le 0$, called \textbf{preparation of the state}, starts at some 
$t=-r$, with $r>0$, when the state $\omega_{-r}$ is invariant under the automorphism $\tau_{-r}$, and ends at $t=0$, when $\omega_{0}$ is
invariant under the automorphism $\tau_{0}$; we require the cyclicity condition
\begin{equation}
\label{(2.35)}
\tau_{-r}=\tau_{0}=\tau^{1}
\end{equation}
It is assumed to be \textbf{nontrivial}, i.e., the state after preparation is \textbf{not} invariant under the evolution:
\begin{equation}
\label{(2.36)}
(\omega_{0} \circ \tau_{t} \ne \omega_{0} \mbox{ for } t > 0
\end{equation}
\end{definition}

\begin{remark}
\label{Remark 2.4.2}
\textbf{Sudden perturbations}
If
\begin{equation}
\label{(2.37)}
\tau^{1} = \tau
\end{equation}
the automorphism is \textbf{continuous} at $t=-r$ and $t=0$. The case 
\begin{equation}
\label{(2.38)}
\tau^{1} \ne \tau
\end{equation}
is a mathematically precise definition of a \textbf{sudden interaction}: it represents a ``finite discontinuity'' which replaces the
``$\delta(t)$- infinity'' of the lifting of barriers in constrained equilibria (\cite{Pa}), such
as the barrier model of remark ~\ref{Remark 2.4.1}. In the examples, however, the infinity appears in a ``coupling constant'' associated to
the perturbation.
\end{remark}

\begin{remark}
\label{Remark 2.4.3}
Note that the automorphism $\tau$ in definition ~\ref{Definition 2.4.2} refers to a \emph{time-independent} interaction. Between $t=-r$ ant $t=0$ 
in the above definition a ``time-dependent Hamiltonian'' is supposed to act. Indeed, according to Definition ~\ref{Definition 2.4.2}, the system 
is closed from $t=0$ to any $t>0$, but not from $t=-r$ to $t=0$, where it is subject to \emph{external conditions}, but is still thermally isolated. 
As informally observed in (\cite{LY}, p.23): ``Hence, as far as the system is concerned, all the interaction it has with the external world during
an adiabatic process can be thought of as being accomplished by means of some mechanical or electrical devices''. This remark by Lieb and Yngvason 
agrees with the intuition which guided our definition ~\ref{Definition 2.4.2}: the ``preparation part'' of the adiabatic transformation thus defined
represents an interaction with the environment of the form preconised by them. The case of sudden perturbations ~\eqref{(2.38)} also illustrates
their remarks on p.17: ``The word 'adiabatic' is sometimes used to mean 'slow' or 'quasi-static', but nothing of the sort is meant here. Indeed,an
adiabatic process can be quite violent. The explosion of a bomb in a closed container is an adiabatic process''. Sudden interactions play an 
important role in our proof of the modified second law in section 3.
\end{remark}

One last, but very important property of an adiabatic transformation is

\begin{theorem}
\label{th:2.4.2}
\textbf{time-arrow theorem}
The preparation part of a system undergoing an adiabatic transformation as defined in ~\ref{Definition 2.4.2} necessarily breaks time-reversal
invariance, that is, there exists some scalar function of the states and observables of the system which is not invariant under time-reversal.

\begin{proof}
By ~\eqref{(2.36)}, there exists some observable $\tilde{A} \in {\cal A}$ such that the function
\begin{equation}
\label{(2.39)}
t \to f_{\omega_{0}}(t,\tilde{A}) \equiv \omega_{0}((\tau_{t}-\mathbf{1})\tilde{A})
\end{equation}
is such that
\begin{equation}
\label{(2.40)}
f_{\omega_{0}}(t,\tilde{A}) \ne 0 \mbox{ for } t>0
\end{equation}
but, by invariance of $\omega_{0}$ under $\tau_{0}$, $\forall A \in {\cal A}$,
\begin{equation}
\label{(2.41)}
f_{\omega_{0}}(0,\tilde{A}) = 0
\end{equation}
Therefore, the preparation of the system undergoing an adiabatic transformation as defined in ~\ref{Definition 2.4.2} breaks 
time-reversal invariance.
\end{proof}
\end{theorem}

\subsubsection{Quantum chaotic states: the role of dynamics and the classical character}

The analog of ``chaotic state'' in section 1 for quantum systems may hold, rather surprisingly, for a class of examples of qss. Take, for
instance, in the weak*-topology
\begin{equation}
\label{(2.42)}
\omega_{1,2}^{\pm} = \lim_{\Lambda \nearrow \mathbf{Z}} \omega_{\Lambda,1,2}^{\pm}
\end{equation}
where
\begin{equation}
\label{(2.43)}
\omega_{\Lambda,1,2}^{\pm} \equiv \otimes_{x \in \Lambda}|\pm)_{x,1,2} \mbox{ with } \sigma_{x}^{1,2}|\pm)_{x,1,2} = \pm |\pm)_{x,1,2} 
\end{equation}
Then, by (\cite{Ra}, (6), p.2950), for the Ising model in dimension $\nu=1$,
\begin{equation}
\label{(2.44)}
\omega_{1,2}^{\pm}(\tau_{t}(\sigma_{x}^{1,2}) = \prod_{j \in \mathbf{N}} \cos^{2}(2 \epsilon(|j|)t)
\end{equation}
For the Dyson model, 
\begin{equation}
\label{(2.45)}
\epsilon(|j|) = \frac{1}{2(|j|)^{s}} \mbox{ with } j \ne 0, s>1
\end{equation}
we find
\begin{equation}
\label{(2.46)}
\omega_{1,2}^{\pm}(\tau_{t}(\sigma_{x}^{1,2})= Cl_{1/2;s>1}^{2}
\end{equation}
where $Cl$ denotes the Cloitre function introduced by Albert and Kiessling \cite{AK}, given by
\begin{equation}
\label{(2.47)}
Cl_{p;s}(t) \equiv \prod_{n \in \mathbf{N}} (1-p+p\cos(n^{-s}t))
\end{equation}
The strong chaotic fluctuations of $Cl_{p;s}(t)$ for large $t$ are given in the figure below, which is Fig.2 of \cite{AK}, here
reproduced with the kind permission of Michael Kiessling:

\includegraphics[height=10cm,width=15cm]{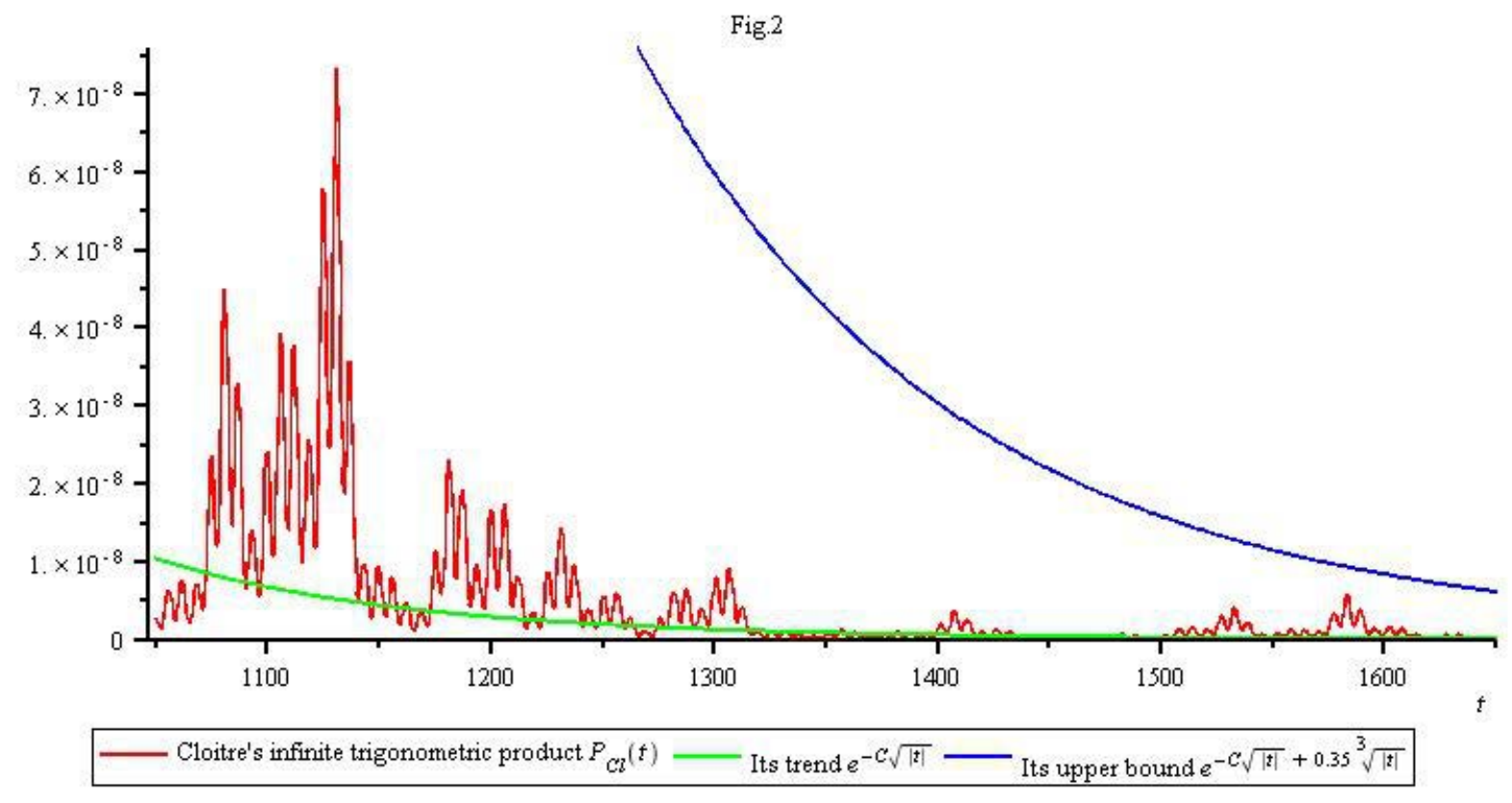}

It may be expected from the figure that, for times $t=1, \cdots, n$, \textbf{exponential sensitivity to initial conditions} holds, i.e., 
$$
|\omega(\tau_{1}(\sigma_{x}^{1})) - \omega(\tau_{1+2^{-n}}(\sigma_{x}^{1}))| = O(2^{-n}) \mbox{ but }
|\omega(\tau_{2^{n}}(\sigma_{x}^{1})) - \omega(\tau_{1+2^{n}}(\sigma_{x}^{1}))| = O(1)
$$
The above means that any initial imprecision in the knowledge of initial conditions becomes exponentially amplified in the course of time.
We are far, however, from being able to \textbf{prove} this conjecture! The above represents, nevertheless, strong graphical evidence 
that \textbf{quantum chaos} may exist for states of infinite systems!

It is important to observe that quantum chaos is restricted to a given universality class of interacting systems. The exponential model
~\eqref{(2.10)} has trivial, soluble dynamics, with a decay of the function $\omega_{1,2}^{\pm}(\tau_{t}(\sigma_{x}^{1,2}) = \sin(t)/t$
\cite{Em}, while the corresponding functions in the Dyson models display the behavior depicted in the present section. Given that $\nu=1$,
the exponential model displays no phase transition, while the Dyson models do (as first shown by Dyson in \cite{Dy}). Therefore, the
mechanism of approach to equilibrium depends on the universality class, see the epilogue.

For $\nu=1$ and only nearest-neighbor interactions, the Ising model is equivalent to a \textbf{Fermi lattice gas}: the interaction contains 
a term
\begin{equation}
 \label{(2.48)}
H_{\Lambda,I} \equiv -1/2 \sum_{x \in \Lambda} a_{x}^{*}a_{x+1}^{*}a_{x+1}a_{x}
\end{equation}
where $a_{x}, a_{x}^{*}$ are operators satisfying mixed commutation and anticommutation relations $[a_{x},a_{x}^{*}]=0$ if $x \ne y$, but
${a_{x},a_{x}^{*}}=1$ (\cite{BRo2}, p.426). In contrast, the ``XY part'' is a quadratic form in the creation-annihilation operators.

The chaotic nature of the dynamics does vary strongly with certain parameters, \textbf{particularly with the time}. The reader should compare 
Fig.2 of the paper by Albert and Kiessling, reproduced in the present paper, with their Fig.1, for small times, which displays
significantly less chaos. We may understand this phenomenon intuitively if we recall that chaotic dynamics is a generic aspect of
classical mechanics (classical dynamical systems), and we expect that the classical limit is defined by a condition of the form
\begin{equation}
\label{(2.49)}
\frac{\hbar}{ET} \ll 1
\end{equation}
where $E$ denotes an energy value in the typical range allowed by the system, and $T$ a corresponding typical time value: the quantity
$\frac{\hbar}{ET}$ is of course dimensionless. In our case, we may take in the Dyson model $E$ as the maximum value in ~\eqref{(2.45)}
and $T$ as the time (see ~\eqref{(2.44)}). In this sense, comparison between the two figures in \cite{AK} also  provides strong graphical
evidence of the dictum: the limit of large times is classical!

\section{The second law as a theorem for two classes of models of interacting spins at $T=0$: the necessity of extended perturbations}

In this section we prove the modified second law (Clausius) for two classes of models. In order to do so, we must add some definitions.

A \textbf{factor} or \textbf{primary} state is defined by the condition that the center
\begin{equation}
\label{(3.1)}
Z_{\omega}= \{\lambda \mathbf{1} \}
\end{equation}
with $\lambda \in \mathbf{C}$.

For quantum spin systems the center $Z_{\omega_{\beta}}$ coincides  with the so called algebra at infinity
$\zeta_{\omega}^{\perp}$, which corresponds to operations which can be made outside any bounded set. As a typical example of an observable
in $\zeta_{\omega}^{\perp}$, let $\omega$ be any translation invariant state. Then the space average of A
\begin{equation}
\label{(3.2)}
\eta_{\omega}(A) \equiv s-lim_{\Lambda \nearrow \infty} \frac{1}{|\Lambda|} \sum_{x \in \Lambda} \Pi_{\omega}(\tau_{x}(A))
\end{equation}
exists, and, if $\omega$ is ergodic, then 
\begin{equation}
\label{(3.3)}
\eta_{\omega}(A) = \omega(A) \mathbf{1}
\end{equation}
This corresponds to ``freezing''
the observables at infinity to their expectation values. The following definition is abstracted from \cite{He}, before his Lemma 1.

\begin{definition}
\label{Definition 3.1}
Two states $\omega_{1}$ and $\omega_{2}$ are \textbf{disjoint} if no subrepresentation of
of $\Pi_{\omega_{1}}$ is unitarily equivalent to any subrepresentation of $\Pi_{\omega_{2}}$ - otherwise they are called \textbf{coherent}.
Two states which induce disjoint representations are said to be disjoint.
\end{definition}

For finite-dimensional matrix algebras (with trivial center) all representations are coherent, and factor representations as well.

We have (\cite{He}, Lemma 6): Let $\omega_{1}$ and $\omega_{2}$ be extremal invariant (ergodic) states with respect to space
translations. If, for some $A \in {\cal A}$,
\begin{equation}
\label{(3.4)}
\eta_{\omega_{1}} (A) = a_{1} \mbox{ and } \eta_{\omega_{2}}(A) = a_{2} \mbox{ with } a_{1} \ne a_{2}
\end{equation}
then $\omega_{1}$ and $\omega_{2}$ are disjoint.

\begin{definition}
\label{Definition 3.2}
We say that a perturbation $V$ is \textbf{extended} iff it is a functional of the elements of all finite parts of $\mathbf{Z}^{\nu}$. It
is termed \textbf{local} iff it is of the form $V=a(\Lambda_{0})$, where $a(\Lambda_{0}) \subset {\cal A}_{\Lambda_{0}}$, with $\Lambda_{0}$
a fixed subset of $\mathbf{Z}^{\nu}$.
\end{definition}

\begin{theorem}
\label{th:3.1}
Let
\begin{equation}
\label{(3.5)}
H_{\Lambda} = \frac{1}{2} \sum_{x,y \in \Lambda} J(|x-y|) \sigma_{x}^{3}\sigma_{y}^{3} \mbox{ with } J(|x|) < 0
\end{equation}
denote the generalized ferromagnetic Ising model (gIm), and $\tau$ the corresponding automorphism corresponding to the infinite system. Then, 
for both the exponential model ~\eqref{(2.10)}, and the Dyson model ~\eqref{(2.11)} with $\alpha > 2$, there exists an infinite number of 
distinct adiabatic transformations satisfying ~\eqref{(2.38)}, i.e., associated to a class of sudden perturbations, each of extended 
type, such that modified second law (Clausius) for $T=0$ holds. For perturbations of local type of the ferromagnetic anisotropic Heisenberg model 
~\eqref{(2.8)}, the modified second law (Clausius) does not hold.

\begin{proof}

Let $\tau$ denote the time-translation automorphism associated to the finite-region Hamiltonian ~\eqref{(3.5)}. In the step of preparation of
the state in definition ~\ref{Definition 2.4.2} of adiabatic transformation, we take
\begin{equation}
\label{(3.6)}
\tau_{-r} = \tau_{0} = \tau_{1} \ne \tau
\end{equation}
That is to say, ~\eqref{(2.38)} is satisfied, characterizing a sudden interaction, where $\tau_{1}$ denotes the automorphism defined by
a finite-volume Hamiltonian 
\begin{equation}
\label{(3.7)}
H(\lambda, \Lambda) \equiv H_{\Lambda}+\lambda V_{\Lambda}
\end{equation}
Above, $V_{\Lambda}$ is a perturbation, and $\lambda$ a coupling constant.
We take the initial state $\omega_{in}=\omega_{0}^{+}$ as one of the ground states of the Hamiltonian defined by ~\eqref{(3.5)}, the
other case, $\omega_{0}^{-}$, as well as any convex combination of the two, being similar. Choose the perturbation $V_{\Lambda}$ as a
tranverse, sufficiently strong magnetic field 
\begin{equation}
\label{(3.8)}
V_{\Lambda} \equiv \sum_{k=-\Lambda}^{\Lambda} V_{k}
\end{equation}
where
\begin{equation}
\label{(3.9)}
V_{k} \equiv 2(-\sum_{x \in A_{k}} h_{x}\sigma_{x}^{1} + \sum_{x \in B_{k}} h_{x} \sigma_{x}^{1})
\end{equation}
with $h_{x} ; x \in \mathbf{Z}$ a sequence of positive numbers satisfying the condition
\begin{equation}
\label{(3.10)}
h_{x} > \sup_{x \in \mathbf{Z}; x \ne 0} (\psi^{-|x|}, |x|^{-\alpha})
\end{equation}
where $\alpha > 2$ and $\psi > 1$. Above,
\begin{equation}
\label{(3.11)}
A_{k} \equiv [kC+1, \alpha_{k}] \mbox{ and } B_{k} \equiv [\alpha_{k}+1, (k+1)C]
\end{equation}
where
\begin{equation}
\label{(3.12)}
\alpha_{k} \equiv \frac{C}{2}(1+f) \mbox{ and } \beta_{k} \equiv \frac{C}{2}(1-f)
\end{equation}
Above, 
\begin{equation}
\label{(3.13)}
0 < f \equiv \frac{A}{B} < 1
\end{equation}
and $C$ is chosen such that
\begin{equation}
\label{(3.14)}
C \ge 2B
\end{equation}
If $A$ and $B$ are chosen as positive integers, $f$ ranges over all rational numbers between zero and one. 

With our choice $T=0$, the restriction of the initial ground state to a region $\Lambda$, $\omega_{\Lambda,0}^{+}$, shifts under conditions
above (~\eqref{(3.7)} et seq.), \textbf{in the limit of infinite coupling} 
\begin{equation}
\label{(3.15)}
\lambda \to \infty
\end{equation}
to the ground state of $V_{\Lambda}$, which is the vector state $(\Psi_{\Lambda}^{f}, \cdot \Psi_{\Lambda}^{f})$,
with
\begin{equation}
\label{(3.16)}
\Psi_{\Lambda}^{f} = \otimes_{k=-\Lambda}^{\Lambda} \Psi_{k}^{1} \otimes \Psi_{k}^{2}
\end{equation}
where
$$
 \Psi_{k}^{1} \equiv \otimes_{x \in A_{k}} |+)_{x}^{1} \otimes \cdots |+)_{x}^{1}
$$
and
$$
\Psi_{k}^{2} \equiv \otimes_{x \in B_{k}} |-)_{x}^{1} \otimes \cdots |-)_{x}^{1}
$$

In fact, the limit of infinite coupling is precisely defined by the above conditions.

Above, by definition, $\sigma_{x}^{1}|\pm) = \pm \frac{1}{2} |\pm)$, and the sets $A_{k}$ and $B_{k}$  satisfy $A_{k} \cup B_{k} = [kC+1,(k+1)C]$. 
Moreover, it follows that $C>1$. The limiting state
\begin{equation}
\label{(3.17)}
\omega_{f}^{1} \equiv \omega_{0}^{+} \circ \tau_{1} \equiv \lim_{\Lambda \nearrow \mathbf{Z}} (\Psi_{\Lambda}^{f}, . \Psi_{\Lambda}^{f})
\end{equation}
satisfies the crucial condition ~\eqref{(2.36)}, i.e., it is not a convex combination of the two ground states $\omega_{0}^{\pm}$. 

In order to prove this, we need only show that $\omega_{f}^{1}(A) \ne \alpha \omega_{0}^{+}(A) + \beta \omega_{0}^{-}(A)$ for all $\alpha$
and $\beta$, and some observable $A$. Indeed, take $A= A_{0} = \sigma_{x}^{1}$ for some fixed $x \in \mathbf{Z}$. Then, by the above construction
it follows that $\omega_{f}^{1}(A_{0}) \ne 0$, but $\alpha \omega_{0}^{+}(A_{0}) + \beta \omega_{0}^{-}(A_{0}) = 0$ for all $\alpha$, $\beta$.

In fact, $\omega_{f}^{1}$ is a pure state, characterized by a rational number $f$ satisfying ~\eqref{(3.13)}. Any two of these pure states 
are disjoint by ~\eqref{(3.2)} and Hepp's lemma ~\eqref{(3.4)}, with the analogue of $\eta_{\omega}$ in ~\eqref{(3.4)} defined by
\begin{equation}
\label{(3.18)}
\eta_{\omega}(\sigma^{1}) \equiv s-lim_{\Lambda \nearrow \mathbf{Z}} \sum_{x \in \Lambda} \frac{\sigma_{x}^{1}}{|\Lambda|}
\end{equation}
The values of the ``macroscopic pointer positions'' are, in the present case, the different rationals satisfying ~\eqref{(3.12)}. We have
thus a countable infinity of distinct pure states (being disjoint, they cannot be unitarily equivalent, and thus not equal, two by two),
parametrized by all rationals between zero and one. They satisfy the assumption of proposition ~\ref{prop:2.4.1} concerning the limit state
being the tracial state, by theorem ~\ref{th:2.1}. The modified second law (Clausius) then follows from proposition ~\ref{prop:2.4.1}. Finally,
the second assertion of the theorem follows from \cite{BKR}, because the anisotropic ferromagnetic Heisenberg Hamiltonian ~\eqref{(2.8)} has
a gap.
\end{proof}
\end{theorem}

\section{Epilogue: Why does the mean entropy increase? The Universe, the time-arrow, AQFT and Earman's principle}

\subsection{Why does the mean entropy increase? The role of dynamics and universality classes in selecting the physical mechanism of approach to a NESS}

We have seen that, when the mean Gibbs-von Neumann entropy is regarded as a functional of the states of infinite systems, it may increase with time
under an adiabatic transformation to a maximum value. The non-existence of a time-arrow for finite systems (Schr\"{o}dinger's paradox, see \cite{Leb})
remains true for states of infinite systems (proposition ~\ref{prop:2.2}).

There are at least two ways to introduce a time-arrow, reviewed in section ~\eqref{(2.3)}: through the Boltzmann equation and evolution, or  
through the preparation part of the adiabatic transformations of section 2.4.2, which improves on the versions of \cite{Wre} 
and \cite{Pe} through the inclusion of sudden interactions.

Theorem ~\ref{th:3.1} treats two different classes of ferromagnetic models, the exponential ~\eqref{(2.10)} and the Dyson model ~\eqref{(2.11)}, which
lie in different universality classes associated to the theory of phase transitions and the critical point in statistical mechanics: given a value
of the dimension (one) and a type of symmetry (discrete), they differ in the type of decay of the interaction potential: exponential or polynomial.
In both cases, for $T=0$, there exists a basic structural transition between states (from pure to mixed) in the mathematical limit of infinite time    
(which physically represents an universal time largely exceeding a relaxation time, see also \cite{Wre}). Similarly, in both cases the initial state 
may be any convex combination of a fixed number of pure states, and the final state must be rather special, viz. the tracial state, of maximum mean
entropy, which comprises with Boltzmann's dictum of the ``maximal wealth of microstates'' (see proposition 2.3). 

There exists, however, a profound difference in the \textbf{mechanisms} of approach to the final NESS: in the exponential model the dynamics
is exactly soluble (with a decay of suitable expectation values of type $\frac{\sin t}{t}$ \cite{Em}, while, in the Dyson case, it is not exactly
soluble and there exists strong graphical evidence that it displays the phenomenon of quantum chaos (section 2.4.3), as a consequence of
the results of Albert and Kiessling \cite{AK} on the Cloitre function for large arguments (time-values). We propose to associate this phenomenon 
to \textbf{Krylov's mechanism} \cite{Pe} of successive ``defocalizing shocks'' beginning at the position of the removed barrier. Since this 
mechanism is believed to be the physical mechanism of approach to equilibrium, we are confident that the present framework may be, in certain
cases, physically relevant. In this connection, we should mention a fourth important aspect characterizing universality classes: interacting 
and non-interacting systems (see, in this respect, \cite{MRW}). We have seen that the Ising model may be interpreted as an interacting Fermi
lattice gas (see ~\eqref{(2.48)}). The XY model, on the other hand, is a quadratic form in the fermion creation and annihilation operators in
this same picture, and may, in this sense, be called ``non-interacting''. Although the XY-model dynamics is not chaotic, the corresponding process 
of ``return to equilibrium'' is rather subtle, and a complete analysis of it may be found in \cite{HuRo}.

\subsection{The Universe and the time-arrow, and the role of quantum fields}

With our definition of the adiabatic transformation in section 2.4.2, we succeeded in decoupling the validity of the second law from the
nature of the initial state of the Universe. Boltzmann, however, thought that a reference to the initial state was inevitable (see \cite{Leb} and
references given there). His conjecture was, however, ultimately correct, because the problem arises as soon as we attempt to apply proposition
2.3 to the Universe. The first natural question would be: why do so many adiabatic transformations occur in the Universe?
Since, however, adiabatic transformations break time-reversal invariance (theorem 2.3), Boltzmann's previous conclusion suggests itself
naturally. And, indeed, the simplest model of the Universe - the ``hot big-bang theory'' ( \cite{Lu})- confirms it, because it leads to an 
ever-expanding Universe, which is one of the phenomena associated to the ``arrow of time''. The details of the expansion depend, however, on
much more refined models. One of them, due to Maeder \cite{Ma}, explores an additional symmetry - the scale-invariance of the macroscopic
empty spaces - which does not require the assumptions of ``dark energy'' or ``dark matter''.

A preliminary model of the CMB (``cosmic background radiation'') which fills the Universe today with a black-body spectrum of temperature
$T \approx 2,7K$ assumes a system of free photons, i.e., which do not interact with the Robertson-Walker metric \cite{Lu}. It requires only
free quantum fields. In general, interacting quantum fields - in principle, described by algebraic quantum field theory (AQFT) \cite{CH} - 
will be required. An important difference between our ferromagnetic models and AQFT is the fact that, in contrast to AQFT, the separating 
character of local algebras does not hold, because our ground states posess annihilators among the local algebras ${\cal A}_{\Lambda}$, viz.
$$
\sum_{x \in \Lambda} (\sigma_{x}^{3} \pm 1/2)
$$  
For this reason, the deep theorem in \cite{CH} that the set of entangled states is generic, is not necessarily true for infinite quantum
ferromagnets.

The issue of the Reeh-Schlieder property in curved space-time is a deep and difficult one: see \cite{Sa}.

\subsection{Some problems posed by AQFT}

It would be very nice to extend the present framework to AQFT, but there are several problems in this connection. Firstly, one may 
conjecture that physical theories such as quantum electrodynamics may be ultimately described by AQFT, but, unfortunately, as observed by 
Lieb and Loss \cite{LL}, the quantity $E_{0}(\Lambda, \Gamma)$, where $\Lambda$ denotes a space cutoff and $\Gamma$ an ultraviolet
cutoff, which is the ground state energy to be subtracted from the Hamiltonian $H_{\Lambda, \Gamma}$, to account for the condition of stability
~\eqref{(2.32)}, is not correctly described by  perturbation theory, which, in the mathematically rigorous approach of ``formal series'', predicts
$O(|\Gamma|^{2})$, contradicting the Lieb-Loss inequalities (see \cite{LL} and references given there). This means that any progress depends
on \textbf{nonperturbative} exact or rigorous results on the dynamics, which are not yet available.

A different approach to the ``time-arrow'' in AQFT is due to Buchholz and Fredenhahen, see their review in arXiv 2305.11709, and references given
there. We agree with their view to the extent that we relate the ``time-arrow'' to \textbf{our} version of the second law, whis is \textbf{not}
of statistical nature, because our time-evolution is deterministic, and the ``non-unitarity'' is a direct consequence of the preparation part of an
adiabatic transformation. The statistical character is a consequence of the whole approach, and is expressed by Proposition 2.3.   

\subsection{Earman's principle}

In the philosophy of physics, Earman's principle \cite{Ear} states: ``While idealizations are useful and, perhaps, even essential to progress in
physics, a sound principle of interpretation would seem to be that no effect can be counted as a genuine physical effect if it disappears when the 
idealizations are removed''. The present paper is concerned by this principle, and we believe it is very instructive to explain why our results 
are compatible with it. At the same time, we hope to contribute to identify the complex variety of mechanisms which are responsible for this 
compatibility.

One good example concerns the ferromagnetic Ising model with nearest neighbor interactions. For dimensions $\nu \ge 2$ it displays a ferromagnetic
phase transition. In the thermodynamic limit, the susceptibility $\chi(T)$ displays a well-known singularity
\begin{equation}
\label{(4.1)}
\chi(T) \approx (T-T_{c})^{-\gamma} \mbox{ for } T \to T_{c}
\end{equation}
where $\gamma = 7/4$ if $\nu = 2$, and $\gamma = 1,25$ if $\nu = 3$ (see \cite{Hu}, p. 396). On the other hand, for finite volume $|\Lambda|$,
$\chi_{\Lambda}(T)$ is a continuous function of $T$. It happens, however, that, for finite but ``sufficiently large'' $|\Lambda|$ (characteristic
of macroscopic systems), the images of the points $\chi_{\Lambda}(T)$ near $T_{c}$ agree with ~\eqref{(4.1)} to extremely high accuracy, although,
for each finite $|\Lambda|$, $\chi_{\Lambda}(T_{c})$ is finite. Thus: the idealized situation (thermodynamic limit) allows no discontinuity but 
remains extremely close to the concrete physical situation.

In the real physical situation (finite systems) the volume $|\Lambda|$ appears explicitly in the various expressions, and it is, therefore,
necessary to substitute experimental values to arrive at concrete results. The above example suggests that the thermodynamic limit - the idealized
situation in Earman's principle - is a \textbf{universal} (viz. does not involve $|\Lambda|$ explicitly) way to describe the large $|\Lambda|$- behavior
of real physical systems. We conjecture that the same phenomenon occurs in the theory of approach to equilibrium (or to a NESS, non-equilibrium
steady state, as in the main text), but now in the \textbf{large time limit}, i.e., for all times greatly exceeding a characteristic 
\textbf{relaxation time}. As in the theory of critical phenomena, the nature of this limit (in particular, the mechanism of approach to equilibrium)
depends on \textbf{universality classes} of dynamical behavior, being ``quantum chaotic'' in certain interesting cases, as in the Dyson model.

As remarked by Wigner in \cite{Wig}, ``it is not at all natural that natural laws exist, much less that man is able to discover them''. This is
a point which complements that part in Earman's principle which states the conjecture that  idealizations (such as considering states of infinite 
systems) may be \textbf{essential} to understand these laws. Schr\"{o} dinger \cite{Schr} even says that the ``second miracle'' in Wigner's quote 
``may well be beyond human understanding''.

We have seen that it is necessary to consider infinite systems in order to \textbf{define} the large (infinite) time limit. This involves an 
additional difficulty: there are several topologies on the set of states. The argument of \textbf{physical naturalness} in section 2.2 places an 
additional requirement on the choice of topology. It shows that our result is compatible with Earman's principle, but involves a deeper discussion
of the relation between mathematics and physics, akin to the issues discussed by Wigner \cite{Wig} in the paper whose title provides part of the answer. 

One last example of Earman's principle is the limit of infinite coupling ~\eqref{(3.15)}. It is expected that the final non-equilibrium state 
$\omega_{f}$ is ``close'' (in the weak*-topology!) to the ``real'' physical state for ``sufficiently large'' coupling. We are, however, far from
being able to prove this! This difficulty is also present in the original ``barrier models'' as described in the text, and shows yet another aspect
of the above-mentioned conjecture that idealizations may be essential to the progress of physics. Indeed, the barrier model is the \textbf{only} 
nontrivial exactly soluble model in conventional thermodynamics, and students of thermodynamics would finish up with empty information if one 
expected that this unique  model of an irreversible adiabatic expansion turned up to be a fake, i.e., misleading regarding ``real physics''!     
 
\textbf{Acknowledgements}

I should like to thank Elliott Lieb and Heide Narnhofer, who have strongly influenced
the present paper, through various scientific exchanges, very heartily. Heide also helped me to improve the final version with several important
remarks. I am also indebted to the late Derek Robinson, who insisted on the relevance
of the property of affinity, which plays an important role in this paper. Finally, I should like to thank Michael Kiessling for his kind
permission to include Fig.2 of his paper with Albert, which turned out, unexpectedly, to describe the chaotic quantum dynamics of the
whole class of Dyson models.

I should like to thank Pedro L. Ribeiro for several constructive comments, in particular his remark about the Reeh-Schlieder property in curved
space-time, as well as D. Buchholz for his remarks and reference concerning the ``time-arrow'' issue.

Last but not least I should like to thank Klaas Landsman for his very good suggestion of including a discussion of Earman's principle.

\end{document}